\documentclass{aicom2e}

\usepackage{amssymb,amsthm}\usepackage{url}\usepackage{graphicx}
\usepackage{multirow}
\begin{document}
\begin{frontmatter} 

\title{Efficient Lineage for SUM Aggregate Queries}

\author[]{\fnms{Foto} \snm{N. Afrati}%
}
\address{School of Electrical and Computer Engineering\\
National Technical University of Athens\\
15780, Athens, Greece}

\author[]{\fnms{Dimitris} \snm{Fotakis}%
}
\address{School of Electrical and Computer Engineering\\
National Technical University of Athens\\
15780, Athens, Greece}

\author[]{\fnms{Angelos} \snm{Vasilakopoulos}%
}
\address{School of Electrical and Computer Engineering\\
National Technical University of Athens\\
15780, Athens, Greece\\
E-mail: avasilako@gmail.com}

\newcommand{\squishlist}{
\begin{list}{$\bullet$}
{ \setlength{\itemsep}{0pt} \setlength{\parsep}{3pt}
\setlength{\topsep}{3pt} \setlength{\partopsep}{0pt}
\setlength{\leftmargin}{1.5em} \setlength{\labelwidth}{1em}
\setlength{\labelsep}{0.5em} } }

\newcommand{\squishend}{
\end{list}  }

\theoremstyle{remark}
\newtheorem{example}{Example}
\newtheorem{remark}{Remark}
\theoremstyle{plain}
\newtheorem{theorem}{Theorem}
\newtheorem{proposition}{Proposition}
\newtheorem{lemma}{Lemma}
\newtheorem{corollary}{Corollary}
\theoremstyle{definition}
\newtheorem{definition}{Definition}

\begin{abstract}
AI systems typically make decisions and find patterns in data based on the computation of aggregate and specifically   sum functions, expressed as queries, on data's attributes. This computation can become costly or even inefficient when these queries concern the whole or big parts of the data and especially when we are dealing with big data. New types of intelligent analytics require also the explanation of why something happened.
 
In this paper we present a randomised algorithm that constructs a small summary of the data, called Aggregate Lineage, which can approximate well and explain all sums with large values in time that depends only on its size. The size of Aggregate Lineage is practically independent on the size of the original data. Our algorithm does not assume any   knowledge on the set of sum queries to be approximated.

\end{abstract}


\end{frontmatter}

\section{Introduction}\label{sec:intro}

  Big data poses new challenges not only in storage but  in intelligent data analytics as well.   Many organisations have the infrastructure to maintain big structured data and need to find methods to efficiently discover patterns and relationships to derive intelligence~\cite{oracle2,oracle-bigdata}. Thus, it would be desirable to be able to construct out of big data a right representative part that can explain aggregate queries, e.g., why the salaries or the sales of a department are high. 
  
   AI systems typically make decisions  based   
on the value of a function computed on data's attributes.
 Several approaches  have in common the computation of aggregates over the whole or large subsets of the data that helps explain patterns and trends of the data. 
E.g., recommendation systems rank and retrieve items that are more interesting for a specific user by aggregating  existing recommendations~\cite{rec2011}.  
  For another example, collaborative filtering computes  a function which uses aggregates and a sum over  the existing ratings from all  users for each product  in order to predict  the preference of a new user~\cite{collaborate,ers-2008}. User preferences are often described as  queries~\cite{journal-aicom09}, e.g.,  queries that give constraints on item features that need to be satisfied.

Another reason for which
data analytics seek to explain data is for 
data debugging purposes.  {\em Data debugging}, which is the
the process that allows users to find incorrect data~\cite{MeliouGNS11,MusluBM13}, is a research direction that is growing fast. 
Data are  collected by various techniques which, moreover, are unknown to and uncontrolled by the user, thus are often erroneous. Finding which part of the data contains errors is essential for companies and affects a large part of their business.

All these applications call for techniques to explain our data. Aggregation is a significant component in all of them.  In this paper we offer a technique that constructs  a summary of the data  with properties that allow it to be used efficiently to explain much of the data behaviour in aggregate for sums. We refer to this summary as {\em Aggregate  Lineage}, since in most applications it represents the source of an aggregate query\footnote{Lineage used to be referred to as ``explain'' in database papers of the late 80's.}.

Lineage (a.k.a. provenance) keeps track of where data comes from. Lineage has been investigated for  data debugging purposes~\cite{IkedaCFSTW12}. Storing the complete lineage of data can be prohibitively expensive and storage-saving techniques to eliminate or simplify similar  patterns in it are studied in~\cite{ChapmanJR08}. 
For select-project-join SQL queries, 
lineage stores the set of {\em all} tuples that were used to compute a tuple in the answer of the query~\cite{BenjellounSHTW08}. This is natural for select-project-join SQL queries where original attribute values are ``copied'' in attribute values of the answer.
However, in an aggregate query the value of the answer is the result of applying an aggregate function over many numerical attribute values.  When we want to understand why we get an aggregate answer it may no longer be important or feasible to have lineage to point to all contributing original tuples and their values. We would rather want to compute {\em few}  values that can be used to tell us as much as possible  about the origin of the result of an aggregate query. However is this at all possible and if it is what are the limitations?

In this paper we initiate an investigation of such questions and, interestingly, we show that useful and practical solutions exist. In particular, we offer a technique that 
uses randomisation  to compute Aggregate Lineage which is  a small representative sample (it is  more sophisticated than just a simple random  sample) of the data.  This sample has the  property to allow for good approximations of a sum query on ad hoc subsets of data -- we call them {\em test queries}. 
 Test queries are applied to the Aggregate Lineage -- not the whole original data. The test queries which we consider are sum queries with same aggregated attribute conditioned with any grouping attributes depending on which subsets of the data we want to test. 
We give performance guarantees about the quality of the results of the test queries that show the approximation to be good for test queries with large values (i.e., close to the total sum over the whole set of data). 
Our performance guarantees hold, with high probability, for any set of queries, even if the number of queries is exponentially large in the size of the  lineage. The only restriction is that the queries should be oblivious to the actual Aggregate Lineage. This restriction is standard in all previous work on random representative subsets for the evaluation of aggregate queries and is naturally satisfied in virtually all practical applications.
The following example offers a scenario about how Aggregate Lineage can be used in data debugging and demonstrates how some test queries can be defined.

\begin{example}
\label{ex:first}
Suppose that the accounting department of a big company  maintains a database  with a relation $Salaries$ with hundreds of attributes and  millions of tuples. Each tuple in the relation may contain an identifier of an employee stored in attribute $EmplID$, his Department stored in attribute $Department$,  his annual salary  stored in attribute $Sal$ and many more attribute values. Other relations are extracted from this relation, e.g.,
a relation which contains aggregated data such as the total sum of salaries of all employees.
A user is trying to use the second relation for decision making but he finds that the total sum of salaries
is unacceptably high. He does not have easy access to the original relation or he does not want to waste time
to pose time-consuming queries on the original big relation. The error could be caused by several reasons
(duplication of data in a certain time period, incorrect code that computes salaries in a new department).
Thus e.g., if we could find the total sum of salaries for employees in the toy department during 2009,
and see that this is  unreasonably high, still close to the first total sum of all employees' salaries, then we will be able to detect such errors and narrow them
down to small (and controllable) pieces of data.

In order to do that, we need the capability of
posing sum queries restricted to certain parts of the data by using combinations of attributes.
This will help the user understand which piece of data is incorrect. We do not know in advance, however,
which piece of data the user would want to inquire and thus Aggregate Lineage should allow the user
to be able to get good approximated answers to whatever queries he wants to try. There are
billions of such possible queries and hence billions of subsets of data which we want to compute
a good approximation of the summation of salaries. We want Aggregate Lineage to offer this possibility.
\end{example}

We propose to keep as Aggregate Lineage a small relation under the same schema of the original relation.
In order to select which tuples to include,
we use valued-based sampling with repetition, i.e., weighted random sampling where the probability of selecting each tuple is proportional to its value on the summed attribute.
The intuition why this method works is the following.
Larger values contribute more to the sum than smaller ones,
thus we expect that tuples with larger values should be selected more often than tuples with smaller values.
Hence, we could end up with a tuple selected many times in the sample even if it appears only once in the original data.
On the other hand, if there are many tuples with values of moderate size, many of them will be selected in the Aggregate Lineage, so that their total contribution to the approximation of the sum remains
significant.

\subsection{Our contribution}

In our approach Aggregate Lineage is  a small relation with same schema
as the original relation and with the property to offer good approximations to test queries posed on it.

 To present performance guarantees, we build on Alth\"{o}fer's Sparsification Lemma \cite{sparse94}. In \cite{sparse94}, Alth\"{o}fer shows that the result of weighted random sampling over a probability vector is a sparse approximation of the original vector with high probability. This technique has found numerous applications e.g., in the efficient approximation of Nash equilibria for (bi)matrix games \cite{LMM03}, in the very fast computation of approximate solutions in Linear Programming \cite{LY94}, and in selfish network design \cite{FKS12}.

 In this paper, we show for the first time that the techniques of \cite{sparse94} are also useful in the context of sum database queries with lineage.
Our results show that the Aggregate Lineage that we extract
 has the following properties (which we describe in technical terms and prove rigorously in Section~\ref{sec:tec}):
\begin{itemize}
\item Its size is practically independent of the size of the original data.

\item It can be used to approximate well all ``large'' sums  (i.e., with  values close to the total sum),  of the aggregated attribute in time that depends only on  its size, and thus is almost independent of the size of the original data.
\end{itemize}

\section{Computing Aggregate Lineage} \label{sec:com}

\begin{figure*}
\begin{center}
\begin{tabular}{| c |l|}
\hline
$t[A]$ & value of attribute $A$ in tuple $ t$\\
\hline
$S$ & sum of all values over attribute $A$\\
\hline
$p_t$ & probability that Algorithm Comp-Lineage selects tuple $t$\\
\hline
$Fr$ & additional attribute recording the frequency of a tuple in the lineage\\
\hline
$L_{R.A}$ & the lineage relation computed by Algorithm Comp-Lineage wrto attribute $A$\\
\hline
$Q(R.A)$ (Sec. 4) & a sub-sum query computed over original relation $R$ wrto attribute $A$\\
\hline
$Q'(L_{R.A})$ (Sec. 4)& sub-sum query $Q$ computed over the aggregate lineage relation  \\
\hline
$I^Q_R$ (Sec. 4)& set of identifiers of the tuples in relation $R$ that satisfy the predicates in query $Q$\\
\hline
\end{tabular}
\caption{Main symbols used in the paper.}
\label{symbols-table}
\end{center}
\end{figure*}

In this section, we present randomised algorithm Comp-Lineage which computes Aggregate Lineage in one pass over the data and in time linear in the size $n$ of the original  database relation. In Section~\ref{sec:tec} we show  that the output of Comp-Lineage is useful to approximate {\em arbitrary} ad-hoc sum test queries in time independent of $n$. We note that our algorithm is agnostic of the specific sum queries that will be approximated by using its output. 

Suppose that we are given a database with a relation   $R$ with $n$  tuples and 
we are given
a positive integer $b$ which is the number of tuples we have decided to include in the Aggregate Lineage
(in Section~\ref{sec:tec} we will explain how we decide $b$ to give good performance and approximation guarantees).
 Suppose that $A$ is a numerical attribute of $R$ which takes nonnegative values.
  Let $S$ be the sum of values of attribute $A$  over
     all $n$ tuples.
The algorithm essentially is a biased sampling with repetition that selects $b$ tuples from $R$.
Each tuple $t$ has probability to be selected equal to $p_t=t[A]/S$ where $t[A]$ is the value of attribute
$A$ in $t$. It collects initially a bag (a.k.a. multiset and is allowed to have the same element more than once) of tuples (since each tuple may be selected multiple times)
which is turned in a set of tuples by adding an extra attribute $Fr$ (for Frequency) which shows the number
of times this tuple is selected. We denote by $L_{R.A}$
the Aggregate Lineage of relation $R$ with sum attribute
$A$.

\smallskip

{\sc Algorithm Comp-Lineage}

{\small

{\bf Input:} A  relation $R$ with  $n$ tuples  and positive integer $b<n$.

{\bf Output:} An $\mathrm{Aggregate~Lineage}$  relation $L_{R.A}$ with at most $b$ tuples.
\begin{itemize}
\item
Randomly select with repetition one out of the $n$ tuples of $R$ in $b$ trials where each tuple $t$ is selected
with probability $p_t$.
\item Form relation  $L_{R.A}$ by including all tuples selected above and adding an extra attribute $Fr$ to
each tuple to record how many times this tuple was selected.
\end{itemize}

 We can use the techniques of \cite{EfraimidisS06} for weighted random sampling and efficiently implement our algorithm to run in linear time in the size of the input either in a parallel/distributed environment or over data streams.

Table~\ref{symbols-table} summarizes the main symbols used throughout the paper.

 \section{Running Example}

\begin{example}\label{ex:al}

We illustrate  Algorithm Comp-Lineage by applying it to Example~\ref{ex:first}  with $b=8,852$ and presenting the data and the Aggregate Lineage in Figure~\ref{fig:ex0-sparse}. Actually Figure~\ref{fig:ex0-sparse} only shows the value of the aggregated attribute ($Sal$ in our example), the rest of the tuple is not
shown.

\begin{figure*}
\begin{center}
\begin{tabular}{| c |c || c|c|c|c|}
\hline
  $Sal$:  & \# of Tuples & Total \# of Tuples  &  $Fr$ & \# of Tuples with & $Sal$: Values $Fr\cdotp S/b$      \\
 O.V. & in $Salaries$ & in Aggregate Lineage &  &  $Fr$  &  in Aggregate Lineage \\
\hline  \hline
\multirow{9}{*}{$10^{9}$} & \multirow{9}{*}{$100$}& \multirow{9}{*}{$100$} & $3$   & $5$ & $3\cdotp S/b= 4.41\times 10^{8}$   \\
\cline{4-6} & &   & $4$  & $10$ & $4\cdotp S/b=5.87\times10^{8}$\\
\cline{4-6} & & & $5$   & $19$ & $5\cdotp S/b=7.34\times10^{8}$\\
\cline{4-6} & & & $6$   & $14$ & $6\cdotp S/b=8.81\times10^{8}$ \\
\cline{4-6} & & & $7$   & $13$ & $7\cdotp S/b=1.03\times10^{9}$\\
\cline{4-6} & & & $8$  & $15$ &  $8\cdotp S/b=1.17\times10^{9}$ \\
\cline{4-6} & & & $9$  & $8$ & $9\cdotp S/b=1.32\times10^{9}$\\
\cline{4-6} & & & $10$  & $12$ & $10\cdotp S/b=1.47\times10^{9}$\\
\cline{4-6} & & & $11$   & $4$ & $11\cdotp S/b=1.62\times10^{9}$ \\
\hline
\multirow{4}{*}{$10^{8}$} & \multirow{4}{*}{$1,000$} & \multirow{4}{*}{$497$} &  $1$    & $347$  &  $ S/b=1.47\times 10^{8}$ \\
\cline{4-6} & & & $2$   & $123$ & $2\cdotp S/b=2.94\times10^{8}$\\
\cline{4-6} & & & $3$   & $20$ & $3\cdotp S/b=4.41\times10^{8}$ \\
\cline{4-6} & & & $4$   & $7$ & $4\cdotp S/b=5.87\times10^{8}$ \\
\hline
$10^{7}$ & $10,000$ & $681$ & $1$ & $681$ &$ S/b=1.47\times 10^{8}$ \\
\hline
$10^{6}$ & $1,000,000$& $6,809$ & $1$ & $6,809$ & $ S/b=1.47\times 10^{8}$  \\
\hline
$10$ & $1,000$& $0$& $0$ & $0$ & $0$  \\
\hline

\hline
\end{tabular}
\caption{Properties of $\mathrm{Aggregate~Lineage}$ $L_{Salaries.Sal}$ for $b=8,852$.
The first two columns describe the data. The next three columns describe the Aggregate Lineage relation. The last column shows how we use this lineage to compute sub-sums.}
\label{fig:ex0-sparse}
\end{center}
\end{figure*}

The first two columns of Figure~\ref{fig:ex0-sparse} present the data in relation $Salaries$. In order to be able to present many tuples
we have chosen a relation with a few values for attribute $Sal$, actually five (i.e., $10^9, 10^8, 10^7,10^6$ and 10) and their Original Values (O.V.) 
are shown in the first column. The second column shows how many tuples in $Salaries$ have these values
in $Sal$. Thus, it says, e.g., that there are 100 tuples with value in $Sal$ equal to $10^9$, 1,000 tuples
with value in $Sal$ equal to $10^8$ and so on.

The third  column in Figure~\ref{fig:ex0-sparse} shows how many tuples from $Salaries$ with a specific
value in $Sal$ are selected by {\sc Algorithm Comp-Lineage} to be included in the Aggregate Lineage relation.
Thus, e.g., all 100 tuples with $Sal=10^9$ were chosen, only 681 tuples with $Sal=10^7$ were chosen and no tuple with $Sal=10$ was chosen.

In order to represent the Aggregate Lineage relation  $L_{Salaries.Sal}$  in the most demonstrative way, we have chosen to
partition its tuples in blocks  (each block further divided in multiple rows in columns 4, 5 and 6), each block corresponding to one value of $Sal $ in $Salaries$. Thus the first block has 9 rows, the second block has  4 rows and the last three blocks have one row each.
This breaking into blocks gives a visualisation of the characteristics of the algorithm.

The fourth    column stores the extra attribute frequency $Fr$  which tells how many times a certain tuple was
selected by the algorithm and the fifth column stores the number of tuples that were selected so many times.
Thus, e.g., the first row says that 5 tuples were selected 3 times each. The ninth row says that 4 tuples from $Salaries$  were selected 11 times each.

The blocks give us an intuition of the characteristics of the Aggregate Lineage.
 The first block
corresponds to the largest value of $Sal$ and tuples with this value (i.e., $Sal= 10^9$) contributed quite heavily to the lineage - all 100 tuples with $Sal = 10^9$ were selected multiple times.
In more detail, there are 100 tuples with value $Sal= 10^9$. 
Of those tuples, 5 were added in the bag 3 times each, 10 tuples were added in the bag 4 times each, and so on. Thus, by considering these 100 tuples, the Algorithm Comp-Lineage added in the bag  $3\cdotp 5+4\cdotp 10+5\cdotp 19+6\cdotp 14+7\cdotp 13+8\cdotp 15+9\cdotp 8+10\cdotp 12+11\cdotp 4=\,681$ tuples in total.  That is to say, each of those 100 tuples contributed on average $6.81$ to the
bag. When we get a set out of the bag by using frequencies (to avoid repeating a tuple multiple times), then we see that the average frequency per tuple is $6.81$. 
So, from this first block, the $681$ tuples in the bag of Algorithm Comp-Lineage are transformed to a set of $100$ tuples in Aggregate Lineage with average frequency $6.81$.
We can compare it with the average frequency in the second 
block which is 0.681 
(this is $1\cdotp 347 + 2\cdotp  123 +3\cdotp  20 +4\cdotp 7=681$ divided by 1000 tuples) and see that, in the data of our example, each tuple of the first block contributes more  heavily  to the lineage.

As we will explain in more detail later, this shows partly why the lineage is useful for discovering almost accurately sub-sums that
are large compared to the total sum, whereas when a sub-sum is small in comparison, then the lineage
cannot be used to compute it accurately.

The second block did not contribute that heavily but still quite a lot, around half of tuples with
$Sal=10^8$ were selected at least once and quite a few more than once, in total this block contributed 681 tuples in the bag. The third block contributed
moderately. The fourth block is interesting because the value of $Sal$ is very small only $10^6$ but
it contributed quite a lot due to the fact that there are many tuples in $Salaries$ with $Sal=10^6$, thus it contributed
almost 85 percent of the tuples in the Aggregate Lineage.

Finally the last column in the figure shows how much each tuple from the Aggregate Lineage contributes
to the approximation of sub-sums that are computed by the test queries.
The same tuple is added in the Aggregate Lineage several times as recorded in the new attribute $Fr$ and thus,
in order to calculate the contribution of a certain tuple, we multiply  its frequency in $Fr$
by $S/b$. By doing so, some tuples (e.g., the ones in the fourth  block) in our
example of Figure~\ref{fig:ex0-sparse} will contribute much more than their actual value in $Sal$. But this is
to compensate for the tuples with value close to it (same value in our example) that are not
selected to be included in the Aggregate Lineage. In the next section we give the technical details on how Aggregate Lineage can be used in order to approximate sub-sums.

 \end{example}

Note that $\mathrm{Aggregate~Lineage}$  does not assume any knowledge of the query set: 
i.e., we run the random selection of Algorithm Comp-Lineage only once and compute $\mathrm L_{R.A}$ without assuming anything about the queries. Then, this same relation $\mathrm L_{R.A}$ can  be used to make us understand any sub-sum test query, without requiring that the test queries are given beforehand or requiring that the test queries are chosen in any specific fashion (e.g., they do not have to be chosen uniformly at random), as long as the query choice is oblivious to the actual sample computed by $\mathrm{Aggregate~Lineage}$%
\footnote{In technical terms, the queries are posed by an \emph{oblivious adversary}, i.e., an adversary that knows how exactly $\mathrm{Aggregate~Lineage}$ works but does not have access to its random choices. The restriction to oblivious adversaries is standard and unavoidable, since if one knows the actual value of $\mathrm L_{R.A}$, he can construct a query that includes only tuples not belonging to $\mathrm L_{R.A}$, for which no meaningful approximation guarantee would be possible.}.
We first present the theoretical approximation guarantees and then demonstrate how these  guarantees play for debugging  on our running example.

\section{Approximation Guarantees of Test Queries on Aggregate Lineage}
\label{sec:tec}
In this section we prove the theoretical guarantees of $\mathrm{Aggregate~Lineage}$.
 Let $R$ be a relation   with a   nonnegative numerical attribute $A$. We consider SUM queries  that ask for the sum of  attribute's $A$ values over arbitrary subsets of the tuples in  relation $R$.
 We use tuple identifiers in order to succinctly represent subsets of tuples.  Thus, any SUM query   defines a set of tuple identifiers for tuples that satisfy its predicates, hence the following formal definitions:

\begin{definition}[Exact SUM $Q(R.A)$]
\label{def:exact}
Let  $R$ be a database relation.  We attach a   tuple identifier on each tuple of $R$.  We denote by $I_R$ the set of all identifiers in relation $R$.  Given an attribute $A$ in the schema of $R$, we denote by $a_i$ the value of attribute $R.A$ in the tuple with identifier $i$ in $R$.

Let $Q$ be a SUM query over $R.A$. We denote by  $I^Q_R$ the set of  tuple identifiers from $I_R$  for tuples of $R$ that satisfy $Q$'s predicates.

The result of a SUM query, $Q(R.A)$, is the summation
of the values of $R.A$ over the set of tuples with identifiers that appear  in $I^Q_R$, i.e.,
$Q(R.A)=\sum_{i\in I^Q_R} a_i$.

\end{definition}
\begin{definition}[Approximated SUM $Q'(L_{R.A})$]
\label{def:appr}

Let $Q$ be a SUM query  over $R.A$ and let $L_{R.A}$ be  an $Aggregate~~Lineage$.  We attach a   tuple identifier on each tuple of $L_{R.A}$.   We denote by $I_{L}$ the set of all identifiers in $L_{R.A}$. We denote by $I^Q_L$ the set of tuple identifiers from  $I_{L}$ for tuples of $L_{R.A}$ that satisfy $Q$'s predicates (since the set of attributes of $R$ is a subset of the set of attributes of $L_{R.A}$, we have that the predicates of a SUM query  $Q$, expressed on  attributes of $R$, define  $I^Q_{L}$).

 We denote by $f_i$ the value of attribute $L_{R.A}.Fr$ in the tuple with identifier $i$ in $L_{R.A}$.

The approximated result of SUM query $Q$, denoted by  $Q'(L_{R.A})$, is the summation
of the values of $L_{R.A}.Fr $ over the set of tuples with identifiers that appear  in $I^Q_L$ multiplied by $S/b$, i.e.,
$Q'(L_{R.A})=\sum_{i\in I^Q_L} f_i\cdotp S/b$.
\end{definition}

The following theorem provides the performance guarantees for any arbitrary set of $m$ SUM queries computed over the  Aggregate Lineage relation in order to serve as an approximation of the corresponding SUM queries
over the original data.

\def\Exp{\mathrm{I\!E}}
\def\Prob{\mathrm{I\!Pr}}

\begin{theorem}\label{th:sparse}
Let $R$ be a relation with $n$ tuples having nonnegative values $a_1, \ldots, a_n$ on attribute $A$, and let $S = \sum_{i=1}^n a_i$. Then, for any collection of $m$ SUM queries $Q_1(R.A), \ldots, Q_m(R.A)$ (not known to the algorithm), any $p \in (0, 1)$, and any $\epsilon > 0$, the Algorithm Comp-Lineage with input all tuples of $R$ and $b = \lceil \ln{(2m/p)}/(2\epsilon^{2}) \rceil$ derives an $\mathrm{Aggregate~Lineage}$ $L_{R.A}$ such that
\( | Q_j(R.A)-Q'_j(L_{R.A})| \leq  \epsilon S \),
for all $j \in [m]$, with probability at least $1-p$.
\end{theorem}

\begin{proof}
The proof is an adaptation of the proof of Alth\"{o}fer's Sparsification Lemma \cite{sparse94}. For simplicity, we assume, without loss of generality, that the set $I_{R}$ of all tuple identifiers of  $R$ in Definition \ref{def:exact}
  is $I_{R}=\{1,\ldots,n\}$. We define $b$ independent identically distributed random
variables $X_1, \ldots, X_b$, which take each value $i \in [n]$ with probability $a_i / S$. Namely, each random variable $X_i$ corresponds to the outcome of the $i$-th trial of Comp-Lineage. For each tuple $i$, its frequency in the sample is 
$f_i = |\{ k \in [b] : X_k = i\}|$.

Let us fix an arbitrary SUM query $Q_j(R.A)$. For each $k \in [b]$, we let $Y^k_j$ be a random variable that is equal to $1$, if $X_k \in I^{Q_j}_{R}$, and $0$, otherwise. Since the random variable $Y^k_j$ is equal to $1$ with probability $Q_j(R.A)/S$, $\Exp[Y^k_j] = Q_j(R.A)/S$. We observe that the random variables $\{ Y^k_j \}_{k \in [b]}$ are independent, because the random variables $\{ X_k \}_{k \in [b]}$ are independent.
Furthermore, we let $Y_j$ be a random variable defined as
\[ Y_j = \frac{1}{b} \sum_{k = 1}^b Y^k_j \]
By definition, $Y_j = \sum_{i \in I^{Q_j}_{R}} f_i/b = \sum_{i \in I^{Q_j}_{L}} f_i/b$, and thus we have that $Y_j = Q'_j(L_{R.A})/S$, i.e., $Y_j$ is equal to the approximated result of the SUM query divided by $S$. Also, by linearity of expectation, $\Exp[Y_j] = Q_j(R.A)/S$.

Applying the Chernoff-Hoeffding bound, we obtain that for the particular choice of $b$, with probability at least $1-p/m$, the actual value of $Y_j$ differs from its expectation $Q_j(R.A)/S$ by at most $\epsilon$, which implies that $Q'_j(L_{R.A})$ differs from $Q_j(R.A)$ by at most $\epsilon S$. Formally, by the Chernoff-Hoeffding bound
\footnote{We use the following form of the Chernoff-Hoeffding bound
(see \cite{Hoef63}): Let $Y^1, \ldots, Y^b$ be random variables
independently distributed in $[0, 1]$, and let $Y = \frac{1}{b}
\sum_{k=1}^b Y^k$. Then, for all $\epsilon > 0$, $\Prob[|Y - \Exp[Y]| >
\epsilon] \leq 2\exp^{-2\epsilon^2 b}$, where $\exp = 2.71\ldots$ is the
basis of natural logarithms.},
$$
 \Prob[|Q'_j(L_{R.A}) - Q_j(R.A)| > \epsilon S] $$
 $$= \Prob[|Y_j - Q_j(R.A)/S| > \epsilon]  $$
 $$\leq 2\exp^{-2\epsilon^2 b} \leq p/m\,,
$$
where the last inequality follows from the choice of $b$.

Applying the union bound, we obtain that
\[
 \Prob[\exists j \in [m] : |Q'_j(L_{R.A}) - Q_j(R.A)| > \epsilon S] \leq p
\]
which concludes the proof of the lemma.
\end{proof}

\begin{example}
\label{ex:practical}
Suppose, in our running example,
we want to be able to answer with good approximation $m=10^{6}$ queries.
What are the guarantees that the theorem provides?
The original data have  $n\approx10^6$ tuples.
Suppose we select the number of tuples in the Aggregate  Lineage  to
be $b\approx 9000$. Then the theorem says that, by setting $\epsilon=0.04$, we can
compute any of $10^{6}$ arbitrary queries within $0.04 S$ of its real value
with probability $1-10^{-6}$. Thus, if the real exact  value of the query $Q_{1}$
is equal to $Q_{1}(Salaries.Sal)=0.4S=S_{1}$ (remember $S$ is the sum over all tuples of relation $R$)
then the approximation will be $0.04S= 0.04S_{1}/0.4=0.1S_{1}$. 
If for another query $Q_{2}$ we have $Q_{2}(Salaries.Sal)=0.8S=S_{2}$ then the approximation will be $0.05S_{2}$, so, then with high probability
we get an answer that is within a factor of $0.05$ of the actual answer.
\end{example}

{\sl Observations on the practical consequences of Theorem~\ref{th:sparse}}.
Examining closely equation $b=\lceil \ln{(2m/p)}/2\epsilon^{2} \rceil$ which gives us an upper bound of the 
number of tuples in the Aggregate Lineage for $m$ queries and with $p$ and $\epsilon$
guarantees as in its statement, we make the following observations:
\begin{itemize}
\item The value of $b$ depends on $m$ as the logarithm, hence if we go from $m$ to $m^2$ queries,
we only need to multiply $b$ by 2 in order to keep the same performance guarantees. Thus it is reasonable to state that, in many practical cases
the number $m$ of queries that can be approximated well can be as large as
a polynomial on the size of data -- even with coefficient in the order of a few hundreds.
\item The value of $b$ does not depend much on $p$ (again only as in the logarithm) but
it depends mainly on $\epsilon$ which controls the approximation ratio
(the approximation ratio itself is $\epsilon /\rho$ if the query to be computed has a sum
$S'=\rho S$).
\end{itemize}

\section{A debugging scenario }\label{sec:deb}
Here is what a user can do for data debugging when using the Aggregate Lineage we propose.

\begin{itemize}
\item  He computes sub-sums by
 filtering some attributes and possibly specific values for these attributes.
E.g., what was the sum of salaries of employees in the toy department in Spring 2010 and
only for those employees who were hired after 2005. The user devises several such
test queries as he sees appropriate and while he computes them and checks that sub-data
is ok or suspicious, he devised different test queries to suit the situation. E.g., if he observes an
unusually large value, close to the total sum, in the query about employees in the toy department and hired before 2005, then
the rest of the queries he devises stay within this department and within the range until 2005, and tries to narrow down
further the wrong part of data. E.g., now he narrows down to each month or/and to employees that are hired between 2005 and 2007, etc. On the other hand, if he finds the answer satisfactory,
then he announces this part of the data correct, therefore stays outside this sub-data and tries to find some other part of the data that are faulty.
The user uses and poses his test queries over the stored small Aggregate Lineage instead of inefficiently use the  original big relation. 
\end{itemize}
In the following example we show how using Aggregate Lineage to approximate test queries applies to our running example.

\begin{example}\label{ex:deb}
We continue our running Example~\ref{ex:al} where we computed $\mathrm{Aggregate~Lineage}$ $L_{Salaries.Sal}$.
Suppose that we have a SUM test query $Q_{1}$ asking the sum of the salaries of a subset of the employees of the company defined from a subset of $EmpID$'s.
Let this subset  consist    of  $50$  employees with salary $10^{9}$, $5,000$ employees with salary $10^{7}$ (so half of them) and of all $10^{6}$ employees with salary $10^{6}$. We compute the query over $Salaries$ and take the  exact answer $1.1\times 10^{12}$.

 In order to use  $\mathrm{Aggregate~Lineage}$ to understand our data we compute $I^{Q_{1}}_{L}$. The $\mathrm{Aggregate~Lineage}$ has  at most $8,852$ tuples.
The identifiers of $I^{Q_{1}}_{L}$ define  the {\em sub-lineage} of query $Q_{1}$ over $L_{Salaries.Sal}$. The sub-lineage of $Q_{1}$ points to $50$ of the tuples of  $L_{Salaries.Sal}$
with original salaries $10^{9}$ and to all $6,809$ tuples   with original $Sal$ values $10^{6}$
 (cf.  Figure \ref{fig:ex0-sparse}).
It will also point to  some tuples of $L_{Salaries.Sal}$ with $Sal$ values  $10^{7}$:
On average query $Q_{1}$ is applied on half of the $681$ selected in Aggregate Lineage tuples, but in extreme cases it may include all or none of them. For this reason, 
it is a good practice to run the randomised algorithm more than
once and compute a few distinct summaries in order to have better results. 
For instance, we may compute three summaries, use some benchmark sub-queries to decide a distance between 
summaries,  toss the summary which is the more distant and keep one of the others arbitrarily.
Note that it is easy to compute the benchmark queries in one pass through the original data in parallel
with computing the lineage. 

We now use the Aggregate Lineage $L_{Salaries.Sal}$ shown in Figure \ref{fig:ex0-sparse} to approximate the value of the sum answer to $Q_1$.
In one worst case query $Q_{1}$ will include:  the  $50$ tuples     with salaries $10^{9}$ from $L_{Salaries.Sal}$ tuples with the larger frequencies and   all $681$ selected  tuples with salaries $10^{7}$.
The  approximation  $Q'_{1}(L_{Salaries.Sal})$ in this case is
 $(4\cdotp 11+
 12\cdotp 10 + \ldots+681
 +6,809)S/b=
7,935\cdotp S/b=\,1.17\cdotp 10^{12}$. In the other extreme case  $Q_{1}$ includes tuples with the smaller frequencies and none of the selected in Aggregate Lineage tuples
 with salaries $10^{7}$, yielding
the   approximation
 $6,995\cdotp S/b=\,1.03\cdotp 10^{12}$.
 We see that $Q_{1}$ is well approximated. Of course the approximation bounds are not the same for every SUM query -  we presented the guarantees in Section \ref{sec:tec}.

 Another straw man approach would be to select as lineage the $8,852$ tuples with larger salary values. This method will select all $100$ tuples with salaries $10^{9}$, all $1,000$ tuples with salaries $10^{8}$ and the remaining $7,752$ tuples from tules with salaries $10^{7}$.  With this approach, query $Q_{1}$ will be on average approximated with the value $50\cdotp10^{9}+ 3,876\cdotp10^{7} \approx 8.8\times 10^{10}$ because it  loses all the information about all original $10^{6}$ tuples with salaries $10^{6}$ contributing to the sum. On another approach, a simple random sampling of 
$8,852$ tuples will  almost always select all of them from the $10^{6}$ many tuples with salaries $10^{6}$.  Query $Q_{1}$ will then  be approximated with the value $8,852\cdotp10^{6} \approx 8.8\times 10^{9}$. Note, on the other hand, that if all original tuples had the same salaries then our method would coincide with simple random sampling.  

\end{example}

\section{Discussion}

We have focused in our exposition only on a single aggregated attribute (e.g., $Sal$ in our example).
This is done for simplicity. Our ideas can be easily extended to include
more aggregated attributes as long as we are willing to keep a distinct aggregate lineage for each attribute.  E..g., suppose we also had a $Rev$ (for Revenue)  attribute for each employee.  In such a case we keep two lineage relations, one for $Sal$ and one
for $Rev$. The algorithm to compute them can be thought of as a parallel implementation of two copies of
the algorithm Comp-Lineage. We need only one pass through the original data. The only difference
is that now, a) we need the two total sums $S_{Sal}$ and
$S_{Rev} $ and b) for each tuple $t$, we have two probabilities $p^{Sal}_t$ and  $p^{Rev}_t$, the
first to be used for the lineage related to attribute $Sal$ and the second to be used for the 
lineage related to attribute $Rev$. 

Algorithm Comp-Lineage performs a weighted random sampling which selects with replacement  $b$ out of $n$ tuples of $R$ where the weight $w_{i}$ for the tuple with identifier $i$ is equal to the value of attribute $A$ of this tuple. Using $b$ copies (each copy selects a single element) of the weighted random sampling with reservoir algorithm presented in \cite{EfraimidisS06}, we can implement Comp-Lineage in one-pass over $R$, in $O(b n)$ time and $O(b)$ space. This implementation can also be applied to data streams and to settings where the values of $n$ and $S$ are not known in advance.

However the technique in \cite{EfraimidisS06}, 
does not seem to be efficiently parallelizable, at least not in a direct way.
Thus the 
problem of how to efficiently implement our technique in distributed computational environments such as
MapReduce remains open. 
Issues about how to implement sampling in MapReduce are discussed in \cite{Carey12}.
Another open problem is how to apply this technique to evolving data~\cite{Ganti02}. In data
streams, we assume that the sample is to be computed over the entire data. When data  continuously evolve with time, the sample may also change considerably with time. 
The nature of the sample may vary with both the moment at which it is computed 
and with the time horizon over which the user is interested in.
We have not investigated here
how to provide this flexibility.

\section{Comparison with Synopses for Data}

There has been extensive research on approximation techniques for aggregate queries on large databases and data streams. Previous work considers a variety of techniques including random sampling, histograms, multivalued histograms, wavelets and sketches (see e.g., \cite{CormodeGHJ12} and the references therein for details and applications of those methods).
Most of the previous work on histograms, wavelets, and sketches focuses on approximating aggregate queries on a given attribute $A$ for specific subsets of the data that are known when the synopsis is computed (e.g., the synopsis concerns the entire data stream or a particular subset of the database). Thus, such techniques typically lose the correlation between the approximated $A$ values and the original values of other attributes. For the more general case of multiple queries that can be posed over arbitrary sets of attributes and subsets of the data not specified when the synopsis is computed, those techniques typically lead to an exponential (in the number of other attributes involved) increase in the size of the synopsis (see e.g., \cite{DobraGGR02,DobraGGR04}).

In contrast, our approach is far more general and does not focus on approximating queries over specific attributes or subsets of the data. Our algorithm computes a small sample without assuming any knowledge on the set of queries and keeps the association between the sampled $A$ values and all other attributes. Then, we can use the Aggregate Lineage to approximate large-valued sum queries over arbitrary subsets of the data that can be expressed over any set of attributes. The Aggregate Lineage can approximately answer a number of queries exponential in its size. Of course, the queries should be oblivious to the actual Aggregate Lineage (technically, they should be computed by an oblivious adversary), but this technical condition applies to all previously known randomised synopses constructions (see e.g., \cite{CormodeGHJ12}).

\section{Conclusions}
We have presented a method that computes lineage for aggregate queries by applying weighted sampling.
The aggregate lineage can be used to compute arbitrary test aggregate queries on subsets of the original data. 
However the test queries can be computed with good approximation 
only if the result of each test query is large enough with respect to the total sum over all the data. 
The aggregate lineage we compute cannot be used to compute test queries if their result is comparatively small. We give performance guarantees. 

Parallel implementation on frameworks such as MapReduce is not studied here. The naive approach of parallelizing \cite{EfraimidisS06} would have either to transmit a large amount of data to the several compute nodes, or to have a makespan linear in $n$.

The idea of getting a single (possibly weighted) random sample from a large data set and using it for repeated estimations of a given quantity has appeared before in the context of machine learning and statistical estimation. 
Boosting techniques \cite{Schapire03} 
such as bootstrapping \cite{Efron79,Efron93} are used.
In \cite{KleinerTSJ12} BLB is used for the efficient  estimation of bootstrap-based quantities in a distributed computational environment.

\section*{Acknowledgements}
This work was supported by the project Handling Uncertainty in Data Intensive Applications, co-financed by the European Union (European Social Fund - ESF) and Greek national funds, through the Operational Program "Education and Lifelong Learning", under the program THALES.

\bibliographystyle{abbrv}
\bibliography{journal-ipl}

\begin{thebibliography}{10}

\bibitem{oracle2}
{\em Information Management and Big Data}.
\newblock An Oracle White Paper, February 2013.

\bibitem{oracle-bigdata}
{\em Big Data Analytics - Advanced Analytics in Oracle Database}.
\newblock An Oracle White Paper, March 2013.

\bibitem{sparse94}
I.~Alth\"ofer.
\newblock On sparse approximations to randomized strategies and convex
  combinations.
\newblock {\em Linear Algebra and Applications}, 99:339--355, 1994.

\bibitem{BenjellounSHTW08}
O.~Benjelloun, A.~D. Sarma, A.~Y. Halevy, M.~Theobald, and J.~Widom.
\newblock Databases with uncertainty and lineage.
\newblock {\em VLDB J.}, 17(2):243--264, 2008.

\bibitem{collaborate}
J.~S. Breese, D.~Heckerman, and C.~M. Kadie.
\newblock Empirical analysis of predictive algorithms for collaborative
  filtering.
\newblock In {\em UAI}, pages 43--52, 1998.

\bibitem{ChapmanJR08}
A.~Chapman, H.~V. Jagadish, and P.~Ramanan.
\newblock Efficient provenance storage.
\newblock In {\em SIGMOD Conference}, pages 993--1006, 2008.

\bibitem{CormodeGHJ12}
G.~Cormode, M.~N. Garofalakis, P.~J. Haas, and C.~Jermaine.
\newblock Synopses for massive data: Samples, histograms, wavelets, sketches.
\newblock {\em Foundations and Trends in Databases}, 4(1-3):1--294, 2012.

\bibitem{DobraGGR02}
A.~Dobra, M.~Garofalakis, J.~Gehrke, and R.~Rastogi.
\newblock Processing complex aggregate queries over data streams.
\newblock In {\em SIGMOD Conference}, pages 61--72. ACM, 2002.

\bibitem{DobraGGR04}
A.~Dobra, M.~Garofalakis, J.~Gehrke, and R.~Rastogi.
\newblock Sketch-based multi-query processing over data streams.
\newblock In {\em Proc. of the 9th International Conference on Extending
  Databas Technology (EDBT 2004)}, volume 2992 of {\em Lecture Notes in
  Computer Science}, pages 551--568. Springer, 2004.

\bibitem{EfraimidisS06}
P.~Efraimidis and P.~G. Spirakis.
\newblock Weighted random sampling with a reservoir.
\newblock {\em Inf. Process. Lett.}, 97(5):181--185, 2006.

\bibitem{Efron79}
B.~Efron.
\newblock Bootstrap methods: Another look at the jackknife.
\newblock {\em The Annals of Statistics}, 7(1):1--26, 1979.

\bibitem{Efron93}
B.~Efron and R.~J. Tibshirani.
\newblock {\em An Introduction to the Bootstrap}.
\newblock Chapman \& Hall, New York, NY, 1993.

\bibitem{FKS12}
D.~Fotakis, A.~Kaporis, and P.~Spirakis.
\newblock Efficient methods for selfish network design.
\newblock {\em Theoretical Computer Science}, 448:9--20, 2012.

\bibitem{Ganti02}
V.~Ganti and R.~Ramakrishnan.
\newblock {\em Mining and monitoring evolving data}.
\newblock Springer, 2002.

\bibitem{Carey12}
R.~Grover and M.~J. Carey.
\newblock Extending map-reduce for efficient predicate-based sampling.
\newblock In {\em ICDE}, pages 486--497, 2012.

\bibitem{Hoef63}
W.~Hoeffding.
\newblock Probability inequalities for sums of bounded random variables.
\newblock {\em Journal of the American Statistical Association},
  58(301):13--30, 1963.

\bibitem{IkedaCFSTW12}
R.~Ikeda, J.~Cho, C.~Fang, S.~Salihoglu, S.~Torikai, and J.~Widom.
\newblock Provenance-based debugging and drill-down in data-oriented workflows.
\newblock In {\em ICDE}, pages 1249--1252, 2012.

\bibitem{journal-aicom09}
D.~Jannach.
\newblock Fast computation of query relaxations for knowledge-based
  recommenders.
\newblock {\em AI Commun.}, 22(4):235--248, 2009.

\bibitem{ers-2008}
M.~Kagie, M.~van~der Loos, and M.~C. van Wezel.
\newblock Including item characteristics in the probabilistic latent semantic
  analysis model for collaborative filtering.
\newblock {\em AI Commun.}, 22(4):249--265, 2009.

\bibitem{KleinerTSJ12}
A.~Kleiner, A.~Talwalkar, P.~Sarkar, and M.~I. Jordan.
\newblock The big data bootstrap.
\newblock In {\em ICML}, 2012.

\bibitem{LMM03}
R.~Lipton, E.~Markakis, and A.~Mehta.
\newblock Playing large games using simple strategies.
\newblock In {\em Proc. of the4th {ACM} Conference on Electronic Commerce
  (EC~'03)}, pages 36--41, 2003.

\bibitem{LY94}
R.~Lipton and N.~Young.
\newblock Simple strategies for large zero-sum games with applications to
  complexity theory.
\newblock In {\em Proc. of the26th {ACM} Symposium on Theory of Computing
  (STOC~'94)}, pages 734--740, 1994.

\bibitem{MeliouGNS11}
A.~Meliou, W.~Gatterbauer, S.~Nath, and D.~Suciu.
\newblock Tracing data errors with view-conditioned causality.
\newblock In {\em SIGMOD Conference}, pages 505--516, 2011.

\bibitem{MusluBM13}
K.~Muslu, Y.~Brun, and A.~Meliou.
\newblock Data debugging with continuous testing.
\newblock In {\em ESEC/SIGSOFT FSE}, pages 631--634, 2013.

\bibitem{rec2011}
F.~Ricci, L.~Rokach, B.~Shapira, and P.~B. Kantor, editors.
\newblock {\em Recommender Systems Handbook}.
\newblock Springer, 2011.

\bibitem{Schapire03}
R.~Schapire.
\newblock The boosting approach to machine learning: An overview.
\newblock In {\em Nonlinear Estimation and Classification}, 2003.

\end{thebibliography}

\end{document}